\documentclass[a4paper,10pt]{amsart}
\usepackage[T1]{fontenc}
\usepackage[latin1]{inputenc}
\usepackage{amssymb,amsmath,amsthm,amsrefs}

\newtheorem{theorem}{Theorem}
\newtheorem{lemma}{Lemma}

\newtheorem{assrcm*}{\textbf{(RCM)}}
\newtheorem*{assp*}{\textbf{(P) Log-H\"older continuity condition}}
\newtheorem*{assi*}{\textbf{(I) Short-range interaction}}
\newtheorem*{dsknN*}{(\textbf{DS}.$k,n,N$)}
\numberwithin{equation}{section}
\numberwithin{theorem}{section}
\numberwithin{definition}{section}
\numberwithin{lemma}{section}

\DeclareMathOperator{\supp}{supp}
\DeclareMathOperator{\prob}{\mathbb{P}}
\newcommand{\condrcm}{\mathbf{(RCM)}}

\newcommand{\ee}{\mathrm{e}}
\newcommand{\esssup}{\mathrm{ess\ sup}}

\newcommand{\Bone}{\mathbf{1}}

\newcommand{\DZ}{\mathbb{Z}}
\newcommand{\DR}{\mathbb{R}}

\newcommand{\esm}{\mathbb{E}}
\newcommand{\DP}{\mathbb{P}}

\newcommand{\BK}{\mathbf{K}}

\newcommand{\FF}{\mathfrak{F}}
\newcommand{\CX}{\mathcal{X}}
\newcommand{\CK}{\mathcal{K}}
\newcommand{\CF}{\mathcal{F}}
\newcommand{\CN}{\mathcal{N}}

\newcommand{\CL}{\mathcal{L}}

\title[Localization for continuous Anderson models]{Conditional distribution  of the sample mean and localization}
\author{Tr\'esor Ekanga}
\address{Universit\'e Paris Diderot 13 Rue Albert Einstein 75013 Paris France}
\email{ekanga@math.cnrs.fr}

\begin{document}

\begin{abstract}
We prove a new hypothesis on the conditional  distribution of the sample mean  of the fluctuations of an i.i.d. random potential in the Anderson model. The paper extends to uniform probability distribution some earlier work with Gaussian distribution  and the  localization results.
\end{abstract}    

\maketitle
\section{ Introduction}
The paper is devoted to some probability estimates on the fluctuations of the sample  mean  which are used in the proofs of eigenvalue concentration  bounds. In some earlier works \cites{AW09,Chu10a, Chu10b, Weg81} the authors analyzed regular distribution  such as the Gaussian distribution  in order to bound  in probability the resonances effects in the multi-particle multi-scale analysis. Recall that the multi-scale analysis is a recursive method which is used in the framework of the proofs on the spectral and exponential Anderson localization \cites{AW09, AGK09,BK05,CS09,Eka11,Eka19,Eka20,GK13}
and is based on a bound of the concentration of the eigenvalues in finite cubes.

In the present work, we treat the case of uniform probability distribution  of the fluctuations of the sample mean. To do so, we will consider the reduction  to the local analysis in the sample space and prove some more probability bounds  on the conditional distribution  function. The probability estimates on the marginal uniform distribution is then used to obtain some localization results such the pure point spectrum and the strong dynamical localization near the bottom of the spectrum. This complement some previous result in this subject of earlier works in the strong disorder regime.

\section{Gaussian i.i.d. samples}
Given a sample of $N$ i.i.d. random variables with  Gaussian  distribution $\CN(0,1)$, and introduce the sample mean $\xi=\xi_N$ and the fluctuations $\eta_i$ around the mean:

\[
\xi_N=\frac{1}{N} \sum_{i=1}^N \xi_i, \quad \eta_i=\xi-\xi_N\quad i=1,\ldots,N.
\]
Recall that $\xi_N$ is independent from the sigma-algebra  $\FF_{\eta}$ generated by $\{\eta_1,\dots,\eta_N\}$ which are linearly and have rank $N-1$. It follows from the fact that $\eta_i$ are all orthogonal to $\xi_N$ with respect to the standard scalar product in the linear space formed by $X_1,\ldots, X_N$ and given by 
\[
\langle Y,Z\rangle:=\esm\left[YZ\right]
\]

where $Y$ and $Z$ are real linear combinations  of $X_1,\ldots X_N$ (recall: $\esm[X_i]=0$). Thus  the conditional probability distribution  of $\xi_N$ given $\FF_N$ coincides with  the unconditional one, so $\xi_N$ follows the law $\CN(0,N-1)$ therefore $\xi_N$ has a bounded density
\[
P_{\xi}(t)=\frac{\ee^{-1/2 t^2}}{\sqrt{2\pi N-1}}\leq \frac{N^{1/2}}{\sqrt{2\pi}}
\]

Further, for any interval $I\subset\DR$ of length $|I|$, we have 
\begin{equation}\label{eq:xi.normal}
\esssup \prob\left\{\xi_N(\omega)\in |\FF\right\}=\prob\{\xi_N(\omega)\in I\}\\
\leq \frac{N^{1/2}}{\sqrt{2\pi}}|I|
\end{equation}
In this particular case of Gaussian samples the conditional regularity of the sample mean $\xi_N$ given the fluctuations $\CF$ is obtain as shows the following elementary example where the common probability distribution of the sample $X_1, X_2$ is just excellent $X_i$ following the law Unif$([0,1])$. So $X_i$ admit a compactly supported probability density bounded by $1$. Indeed,  set

\[
\xi=\xi_2= \frac{X_1+X_2}{2}, \quad \eta=\eta_1=\frac{X_1-X_2}{2},
\]
The random vector $(X_1,X_2)$ is uniformly distributed in the unit square $[0,1]^2$ and the  condition  $\eta=c$  selects a straight line  in the  two dimensional plane  with coordinates $(X_1,X_2)$, parallel to the main diagonal $\{X_1=X_2\}$. The conditional  distribution  of $\xi$ given $\{\eta=c\}$ is the uniform distribution  on the segment
\[
J_c:=\{(X_1,X_2): x_1-x_2=2c, 0\leq x_1,x_2\leq 1\}
\]
of length vanishing at $2c=\pm 1$. For $|2c|=1$ the conditional distribution  of $\xi$ on $J_c$ is concentrated on a single point. We will discuss in  the next Section  the general  case $N\geq 2$ i.i.d. random variables uniformly distributed in $[0,1]$.

\section{The main applications}     

\subsection{The conditional empirical mean in eigenvectors correlators bounds}
Let $\Lambda\subset \DZ^d$ be a finite  subset with $|\Lambda|=N\geq 1$, and $H_{\Lambda}(\omega)$ be a random discrete Schr\"odinger operators acting in the Hilbert space  $\ell^2(\Lambda)$ with i.i.d. random potential $V:\Lambda\times \Omega\rightarrow \DR$, relative to a probability space  $(\Omega,\FF,\DP)$. Let write  the random field $V$ on $\Lambda$ as 
\[
V(x,\omega)=\xi_N(\omega)+\eta_x(\omega),
\]
we can also write $H(\omega)$ as 
\[
H(\omega)= \xi_N(\omega)\Bone + A(\omega),
\]
where the  operator $A(\omega)$ is $\FF_{\eta}$-measurable and so  are its eigenvalues $\tilde{\mu}_j(\omega)$, $j=1,\ldots,N$. Since $A(\omega)$ commutes with the scalar operator $\xi_N(\omega)\Bone$, the eigenvalues $\lambda_j(\omega)$ of $H(\omega)$ have the form
\begin{equation}\label{eq:lambda}
\lambda_j(\omega)=\xi_N(\omega)+\mu_j(\omega)
\end{equation}
The equation \eqref{eq:lambda} implies the following bound: for any interval $I=[t,t+s]$

\begin{align*}
\prob\left\{tr P_I(H(\omega))\geq 1\right\}& \leq \sum_{j=1}^N \prob\{ \lambda_j(\omega)\in I\}
\sum_{j=1}^N\prob\{\xi_N(\omega)+\mu_j(\omega)\in I\}\\
&=\sum_{j=1}^N \esm\left[\prob\{\xi_N(\omega)+\mu_j(\omega)\in I\mid\FF_{\eta}\}\right]\\
&=\sum_{j=1}^N \esm\left[\prob\{\xi_N(\omega)\in[-\mu_j(\omega)+t, -\mu_j(\omega)+t+s]|\FF\}\right]
\end{align*}
Therefore, we have 

\begin{align*}
\prob\left\{\xi_N(\omega)+\tilde{\mu}_j\in I\mid\FF_{\eta}\right\}&=\prob\left\{\xi_N(\omega)\in\left[\mu_j+t,\mu_j+t+s\right]\mid \FF_{\eta}\right \}\\
&=\prob\left\{\xi_N\in [\tilde{\mu}_j,\tilde{\mu}_j+s]\mid \FF\right\}
\end{align*}
where $\tilde{\mu}_j(\omega)=-\mu_j(\omega)+t$ are $\FF_{\eta}$-measurable. Let us now introduce  the conditional continuity modulus of $\xi_N$ given $\FF_{\eta}$:

\[
\nu_N(s):=\sup_{t\in\DR}\esssup \prob\{\xi_N\in[t,t+s]\mid \FF_{\eta}\}, \quad s\in(0,\infty).
\]
We  have that
\[
\prob\{\lambda_j\in I\mid \FF_{\eta}\}\leq \nu_N(s),
\]
Thus 
\begin{equation} \label{eq:trace.PI}
\prob\left\{ tr P_I(H(\omega))\geq 1\right\}\leq N\nu_N(s)=|\Lambda|\nu_N(s).
\end{equation}
We also consider the probabilities 
\begin{equation}\label{eq:s}
s\longrightarrow \prob\left\{\xi_N(\omega)\in \left[\tilde{\mu}(\omega),\tilde{\mu}(\omega)+s\right]\right\},
\end{equation}  
and 
\[
s\longrightarrow \prob\left(\xi_N(\omega)\in\left[\tilde{\mu}(\omega)+s\right]\mid\FF_{\eta}\right)
\]
where $\tilde{\mu}$ is $\FF_{\eta}$-measurable  

\subsection{The Gaussian case}
In the case where $X_i$ follows the law $\CN(0,1)$, using the bound \eqref{eq:xi.normal}, we have that
\[
\prob\{ tr P_I(H(\omega))\geq 1\}\leq N\cdot\frac{N^{1/2}}{\sqrt{2\pi}}|I|=\frac{|\Lambda|^{3/2}}{\sqrt{2\pi}}|I|
\]

\subsection{Reduction to the local analysis in the sample space}
Assume that the support $S\subset\DR$ of the common continuous marginal  probability measure $\DP_V$ of the i.i.d.  random variables $X_j$, $1\leq j\leq N$, is covered by a finite  or countable union  of intervals 
\[
S\subset \bigcup_{k\in\CK} J_k, \quad \CK\subset\DZ, \quad J_k=[a_k,b_k], \quad a_{k+1}\geq b_k
\]
Let $\BK=\CK^N$ and for each $k=(k_1,\ldots,k_N)\in\BK$ denote $J_k=\prod_{i=1}^N J_{k_i}$
Due to the continuity of the marginal measurable, $J_k$ are disjoint: for all $k\neq \ell$ $\DP_V(J_k\cap J_{\ell})=0$. Respectively the family of the parallelepipeds $\{J_k,k\in \CK\}$ forms a partition  $\CK$ of the sample space, which  we will often identify with the probability space $\Omega$. Furthermore, denote by $\FF_{\CK}$ the sub-sigma algebra  of $\FF$ generated  by the partition $\CK$. Now the quantities of the form \eqref{eq:s} can be assessed as follows

\begin{align*}
\prob\left\{\xi_N\in[\tilde{\mu},\tilde{\mu}+s]\right\}&=\esm\left[\prob\left\{\xi_N\in[\tilde{\mu},\tilde{\mu}+s]\mid \FF_{\CK}\right\}\right]\\
&\leq\sum_{k\in\CK}\prob\{J_k\}\prob\left\{\xi_N\in[\tilde{\mu},\tilde{\mu}+s]\mid J_k\right\}
\end{align*}
Denote by $\DP_k\{\cdot\}$ the conditional probability measure, given $\{X\in J_k\}$ $\esm_k[\cdot]$. The respective  expectation and  $p_k=\DP\{J_k\}$. We have that
\begin{align*}
\prob\left\{ \xi_N\in[\tilde{\mu},\tilde{\mu}+s]\right\}&\leq\sum_{k\in\CK} p_k\esm_k\left[\prob_k(\xi_N\in[\tilde{\mu},\tilde{\mu}+s]\mid \FF_{\eta})\right]\\
&\leq\sup_{k\in\CK} \esm_k\left[\prob_k\{\xi_N\in[\tilde{\mu},\tilde{\mu}]\mid\FF_{\eta}\}\right]\\
&\leq \sup_{k\in\CK} \esm_k\left[\prob_k\{\xi_N\in[\tilde{\mu},\tilde{mu}+s]\mid \FF_{\eta}\}\right]
\end{align*}
We also  give in the next Section  analog estimate in the case  of uniform marginal distribution on the i.i.d. variables $X_i$

\section{Uniform probability distribution}

Let be given a real number $\ell\in(0,\infty)$ and an integer $N\geq 2$. We consider a sample of $N$ i.i.d. random variables  with uniform distribution  Unif$([0,\ell])$  and introduce the sample mean $\xi=\xi_N$ and the fluctuations, $\eta_i$ around the mean 
\[
\xi_N=\frac{1}{N} \sum_{i=1}^N X_i, \quad \eta_i=X_i-\eta_N.
\]
We also need a rescale  empirical mean 
\begin{equation}\label{eq:xi.tilde}
\tilde{\xi}_N= N^{1/2} \xi_N
\end{equation}
for the purposes of the orthogonal transformations  
\[
(X_1,\ldots,X_n)\rightarrow (\tilde{\xi}_N,\tilde{\eta}_2,\ldots,\tilde{\eta}_N)
\]
Next $X_i=\eta_i +N^{-1/2}\tilde{\xi}_N, \quad i=1,\ldots,N$. Further consider the Euclidean space  $\DR^N$ of real linear combinations of the random variables $X_i$ with the scalar product $\langle X',X''\rangle=\esm[X'X'']$. Thus the variables $\eta_i:\DR^N\rightarrow\DR$ are invariant under the group of translations
\[
(X_1,\ldots,X_N)\rightarrow (X_1+t,\ldots+t), \quad t\in\DR
\]
and so are their differences
\[
\eta_i-\eta_j=X_i-X_j, \quad 1\leq i_leq j\leq N.
\]
Introduce  the variables
\[
Y_i=\eta_i-\eta_N, \quad 1\leq i\leq N-1,
\]
Then the space $\DR^N$ is fibered into a union  of affine lines of the form 
\begin{align*}
\tilde{\CX}(Y)&:=\{X\in\DR^N: \eta_i-\eta_N=Y_i, i\leq N-1\}\\
& :=\{ X\in\DR^N: X_i-X_N=Y_i, i\leq N-1\}
\end{align*} 
labeled by the elements $Y=(Y_1,\ldots,Y_{N-1})$ of $Y^{N-1}\cong \DR^{N-1}$. Set 
\[
\CX(Y)= \tilde{X}(Y)\cap C_1=\{X\in C_1: X_i-X_N=Y_i, i\leq N-1\}
\]
and consider each non-empty interval $\CX(Y)\subset\DR^N$ with the natural structure of a probability space from $\DR^N$   
\begin{itemize}
\item[(i)]  if $|\CX(Y)|=0$, then we introduce the trivial sigma-algebra and trivial counting measure.
\item[(ii)] If $|\CX(Y)|=r\in(0, \infty)$ then we use the inherited structure of an interval of a one-dimensional  affine line  and  the normalized measure with constant density $r^{-1}$ with respect to the inherited Lebesgue measure on $\CX(Y)$. 
\end{itemize}

Introduce  an orthogonal coordinate transformation in $\DR^N$, $X\rightarrow (\tilde{\xi}_N,\tilde{\eta}_1,\ldots,\tilde{\eta}_{N-1})$ such that 
\begin{equation}\label{eq:xi.tilde}
\tilde{\xi}_N=N^{-1/2}\sum_{i=1}^N X_i=N^{1/2}\xi_N
\end{equation}
We have from \eqref{eq:xi.tilde} that for any given $a\in\DR$, $s\in(0,\infty)$ and some $a'\in\DR$,
\begin{equation}\label{eq:xi.s}
\xi_N\in[a,a+s]\Longleftrightarrow \tilde{\xi}_N\in[a', a'+N^{1/2}s].
\end{equation}
Denote $J^{(\ell)}=[0,\ell]^N$ and introduce the random variables 
\begin{equation}\label{eq:nu.N}
\nu_N(s, J^{(\ell)})=\nu_N(s; J^{(\ell)},X):=\esssup \sup_{t\in\DR}\prob\{\xi_N\in[t,t+s]\mid \FF_{\eta}\}
\end{equation}
since the $\{X_i\}$ are i.i.d. with uniform distribution on  $[0,\ell]$, the distribution  of the random vector   $X(\omega)$ is uniform in the  cubes  $J^{(\ell)}=[0,\ell]^{N}$, inducing a uniform conditional distribution  on each  element $\CX(Y)$. Thus by \eqref{eq:xi.s} and \eqref{eq:nu.N},
\begin{equation}\label{eq:nu.Jl}
\nu_N(s,J^{(\ell)})=\frac{N^{1/2}s}{\CX(Y)}
\end{equation}

\begin{lemma}\label{lem:prob.X}
Consider the i.i.d. random variables  $X_1,\ldots, X_N$ with $X_i$ following the uniform law Unif$(J_{\ell,i})$ with 
\[
J_{\ell,i}=[a_i,a_i+\ell]\subset \DR,\quad \ell\in(0,\infty)
\]
For any $\delta\in(0,\ell]$
\[
\prob\{|\CX(X)|\leq \delta\}\leq \sum_{i=1}^N\prob\{ X_i-a_i\leq \delta\}.
\]
\end{lemma}

\begin{proof}
We can consider the case  where $a_i=0$, $1\leq i\leq N$ with $X_i$ following the law Unif$([0,\ell])$. Otherwise, we make change  of variables $X_i\rightarrow X_i-a_i$. Let 
\[
\underline{X}=\underline{X}(X)=\min_i X_i.
\]
Observe that, each  $N^{1/2}X_i$, $i=1,\ldots,N$ restricted to $\CX(Y)$ provides a normalized length parameter on $\CX(Y)$. So the range of each $N^{1/2}X_i\mid_{\CX(Y)}$ is an interval  of length $|\CX(Y)|$. One can decrease e.g. the value of $X_1$ as long as all $\{X_i, 1\leq i\leq  N\}$ are strictly positive. Therefore, the maximum decrement of $X_i$ along $\CX(Y)$ is given by $\underset{X}(X)$ so the length of the normalized length parameter $N^{1/2}X_1$ along $\CX(Y(X))$ is an interval of length
\begin{equation}\label{eq:Y.N}
|\CX(Y(X))|\geq N^{1/2}\underline{X}(X)
\end{equation}  
Let
\[
\rm{A}_i(t):=\{X_i\leq t\},\quad \rm{A}(t)=\bigcup_{i=1}^N A_i(t)
\] 
\[
\rm{A}^c(t)=\Omega\setminus A(t)
\]
and note that, by \eqref{eq:Y.N}: 
\[
\min_{X\in \rm{A}^c(t)}|\CX(X)|\geq N^{1/2}\min_{X\in A^c(t)}\underline{X}(X)\geq N^{1/2}t,
\]
Equivalently, setting $u=N^{1/2}t$, so $t=N^{-1/2}u$, we have 
\[
|\CX(X)|\leq u\Longrightarrow X\in (N^{-1/2} u).
\]
With $u=\delta$ we deduce that
\[
\prob\{A(N^{1/2N^{-1/2}}\delta)\}=\prob\{A(\delta)\}\leq \sum_{i=1}^N\prob\{X_i\leq \delta\}
\]
\end{proof}

\begin{theorem}\label{thm:prob.nu}
Consider i.i.d. random variables  $X_1,\ldots,X_n$ with $X_i$ following the uniform law Unif$(J_{\ell,i})$ where $J_{\ell,i}=[a_i,a_i+\ell]\subset\DR$, $\ell\in(0,\infty)$. For any $\delta\in(0,\ell]$,
\[
\prob\{\nu_N(s, J^{(\ell)})\geq \delta^{-1}s\}\leq \frac{N\delta}{\ell}
\]
In particular, with $\delta=s^{\alpha}$
\[
\prob\left\{\nu_N(s,J^{(\ell)})\geq s^{1-\alpha}\right\}\leq N\ell^{-1}s^{\alpha}
\]
\end{theorem}
\begin{proof}
The random variable $X=(X_1,\ldots,X_N)\rightarrow |\CX(Y(X))|$ is $\FF_{\eta}$-measurable and takes constant values $|\CX(Y)|$ on each  element $\CX(Y)$. By \eqref{eq:nu.Jl}, for any $\delta\in(0,\infty)$,
\begin{align}
\prob\{\nu_N(s,J^{(\ell)})\geq \delta^{-1}s\}&\leq \prob\{\frac{N^{1/2}}{|\CX(Y)|}\geq \delta^{-1}s\}\notag\\
&\prob\{|\CX(Y)|\leq N^{1/2}s\}\label{eq:prob.nu.delta}
\end{align}  
Finally the result follows from \eqref{eq:prob.nu.delta} and Lemma \ref{lem:prob.X} since
$X_i$ follows the uniform law Unif$([0,\ell])$ 
\[
\prob\{X_i\leq \delta\}=\ell^{-1}\delta.
\]
\end{proof} 

\section{Some more probability distribution bounds}
We can improve  the bound of Theorem \ref{thm:prob.nu} for the main applications  to the multi-scale analysis:

\begin{lemma}\label{lem:prob.densities}
Assume that the i.i.d. random variables $X_1,\ldots,X_N$, $N\geq 2$ admit  the probability density $p_V$ with $\|P_V\|_{\infty}\leq \overline{\rho}$. Then 
\[
\prob\{|\CX(Y)|\leq r\} \leq \frac{1}{4} \overline{\rho}^2r^2N
\]
In particular, for $X_j$ following the uniform law  Unif$([0,\ell])$, we have that
\[
\prob\{|\CX(Y)|\leq r\}\leq \frac{r^2N}{4\ell^2}
\]
\end{lemma} 
\begin{proof}
Let $\underline{X}=\underline{X}(X)=\min_i X_i, \quad \overline{X}=\overline{X}(X)=\max_i X_i$ while  $\overline{X}(X)$ and $\underline{X}(X)$ vary along the elements $\CX(Y)$, their difference $\overline{X}(X)-\underline{X}(X)$ does not  and it is uniquely  determined by $\CX(Y)$. Each $N^{1/2}X_i$, $i=1,\ldots,N$ restricted  to $\CX(Y)$ provides  a  normalized length parameter on $\CX(Y)$, thus the range of each $N^{1/2}X_i\mid_{\CX(Y)}$ is an interval of length $|\CX(Y)|$
We can increase (resp. decrease) e.g. the value of $X_1$, as long as all $\{X_i, 1\leq i\leq N\}$ are  strictly smaller than $\ell$ (resp.  strictly positive).
Therefore the maximum increment of $X_1$ (Indeed, of any $X_i$) along $\CX(Y)$ is  given by $\ell-\overline{X}(X)$, and its maximum decrement equals $\underline{X}(X)$, so the range of the normalized length parameter $N^{1/2}X_1$ along $\CX(Y(X))$ is an interval of length $N^{1/2}(\ell-\overline{X}(X)+\underline{X}(X))$
\[
|\CX(Y(X))|=N^{1/2}(\ell-\overline{X}(X)+\underline{X}(X))
\]
Since both $\underline{X}(X)$ and $\ell-\overline{X}(X)$ are non-negative 
\[
\underline{X}+(\ell-\overline{X})\leq t \Longrightarrow \max\{\overline{X},\ell-\overline{X}\}\leq t/2
\]
with $0\leq t\leq \ell$, $(\ell-X_i\leq t/2)$ implies 
$(X_i\geq t/2)$, thus denoting  
\[
A_{ij}(t):= \{X_i\leq t/2\}\cap\{\ell-X_j\leq t/2\}
\]
we have that  for any $i$,
\[
A_{ii}(t)=\{X_i\leq t/2\}\cap\{\ell-X_i\leq t/2\}=\emptyset
\]
Therefore, 
\[
\{\max\{\underline{X}(X),\ell-\overline{X}(X)\}\leq t/2\}\cap \{\bigcup_{i\neq j}\{X_i\leq t/2, \ell-X_j\leq t/2\}
\]
Thus the union $\bigcup_{i\neq j}A_{ij}(t)$ contains  all samples $X$ with $|\CX(Y)|\leq t/2$. The sample $\{X_k\}$ is i.i.d. with $X_k$ following the law Unif$([0,\ell])$ so for any $i\neq j$.
\[
\prob\left\{A_{ij}(t)\right\}=\prob\{X_i\leq t/2\}\cdot\prob\{\ell-X_j\leq t/2\}=\frac{t^2}{4\ell^2}
\]
owing to Lemma \ref{lem:prob.X}
\begin{align*}
\prob\{|\CX(Y)|\leq r\}&=\prob\{N^{1/2}((\ell-\overline{X}(X))+\underline{X}(X))\leq r\}\\
&=\prob\{((\ell-\overline{X}(X))+\underline{X}(X))\leq r N^{-1/2}\}\\
&\leq \sum_{i\neq j}\prob\{A_{ij}(rN^{-1/2})\}\leq N(N-1)\frac{(\tilde{p}rN^{-1/2})^2}{4}\\
&\leq \frac{1}{4} \tilde{p}^2r^2N
\end{align*}
\end{proof}   
\begin{theorem}\label{thm:prob.nu.Jl}
Consider the i.i.d. random variables $X_1,\ldots,X_n$ with  each $X_i$ following the uniform law Unif$([0,\ell])$. For any $0\leq \delta\leq s\leq \ell$
\[
\prob\{\nu_N(s;J^{(\ell)})\geq \delta^{-1}s\}\leq  \frac{N^2\delta^2}{4\ell^2}
\]
In particular with $\delta	=s^{\alpha}$ $\alpha\in(0,1)$,
\[
\prob\{\nu_N(s,J^{(\ell)}) \geq s^{1-\alpha}\}\leq \frac{N^2s^{2\alpha}}{4\ell^2}
\]
\end{theorem}  
\begin{proof}
As before, we associate with each point $X\in\DR^N$ the straight line $X\in\CL(Y(X))$ parallel to the vector  $v=(1,\ldots,1)$ and consider their intersections  $\CX(Y(X))=\CL(Y(X))\cap J^{(\ell)}$. By \eqref{eq:nu.N} for any $\delta\in(0,\infty)$
\begin{align}
\prob\{\nu_N(s)\geq \delta\}&\leq \prob\{\frac{N^{1/2}s}{|\CX(Y)|}\geq \delta\}\label{eq:Nu.NS}\\
&\prob\{|\CX(Y)|\leq N^{1/2}s\delta^{-1}\}
\end{align}
Let
\[
\underline{X}=\underline{X}(X)=\min_{i} X_i,\quad \overline{X}(X)=\max_i X_i, \quad \overline{X}(X)=\max_i X_i,
\] 
while $\overline{X}(X)$ and $\underline{X}(X)$ vary along the elements $\CX(Y)$, their difference $\overline{X}(X)-\underline{X}(X)$ does not it. Therefore  
\begin{align*}
&\{\max\{\underline{X}(X), \ell-\overline{X}(X)\}\leq t/2\}\\
&\subset\bigcup_{i\neq j}\{X_i\leq t/2, \ell-X_j\leq t/2\}
\end{align*} 
Thus the union $\bigcup_{i\neq j} A_{ij}(t)$ contains all samples $X$ with $|\CX(Y)|\leq t/2$. The sample $\{X_k\}$ is i.i.d., with $X_k$ following the uniform law Unif$([0,\ell])$, so for any $i\neq j$
\begin{align*}
\prob\{A_{ij}(t)\}&=\prob\{X_i \leq t/2\}\cdot\prob\{\ell-X_j\leq t/2\}\\
&=\frac{t^2}{4\ell^2}
\end{align*} 
Owing to \eqref{eq:Y.N}  
\begin{align*}
\prob\{|\CX(Y)|\leq \}&=\prob\{ N^{1/2}((\ell-\overline{X}(X))+\underline{X}(X))\leq r\}\\
&=\prob\{((\ell-\overline{X}(X))+\underline{X}(X))\leq N^{-1/2}\}\\
&\leq \sum_{i\neq j}\prob\{\rm{A}_{ij}(rN^{-1/2})\}\leq N(N-1)\frac{(rN^{-1/2})^2}{4\ell^2}\\
&\leq \frac{r^2N}{4\ell^2}
\end{align*}
Setting $r=N^{1/2}\delta$, we infer from \eqref{eq:Nu.NS} 
\[
\prob\{\nu(s,\ell)\geq\delta\}\leq \frac{N^2s^2}{4\ell^2}
\]
proving the required results.
\end{proof}

Let $Q\subset\DZ^d$ be a parallelepiped and consider the sample of i.i.d. random variables  $\{V(y,\omega), y\in Q\}$ and consider  the sample mean $\xi_Q$ and the conditional continuity modulus $\nu_{|Q|}(s)$ given the sigma-algebra of fluctuations. We have  the following hypothesis:
\begin{assrcm*}
For some $C',C'',A',A'',B',B''\in(0,\infty)$,
\[
\prob\{\nu_{|Q|}(s)\geq C'|Q|^{A'}s^{B'}\}\leq C''|Q|^{A''}s^{B''},
\]
\end{assrcm*}  
Now for an i.i.d. sample with distribution Unif$([0,\ell])$, $\ell\in(0,\infty)$, Theorem \ref{thm:prob.nu.Jl} can be  reformulated as follows.

\begin{theorem}\label{thm:prob.Nu.Q}
Consider an i.i.d. random field $V:\DZ^d\times \Omega\rightarrow\DR$ with marginal distribution Unif$([c,c+\ell])$, $c\in\DR$.  Then $V$ satisfies the condition $\condrcm$ with the parameters  which can be chosen as follows
\[
C'=1,  \quad A'=0, \quad b'=1-\alpha
\]
\[
C''=\frac{1}{4\ell^2}, \quad A''=2, \quad b''=1-\alpha
\]
We can set for example
\[
b'=b''=2/3
\]
Explicitly
\[
\prob\{\nu_{|Q|}(s,\ell)\geq s^{1-\alpha}\}
\]
\end{theorem} 

\section{Smooth positive probability densities}
We are now ready to consider a richer class of probability distributions with uniform positivity and smoothness of the probability density on a compact interval

\begin{theorem}\label{thm:densities}
Assume that the common probability distribution of the i.i.d. random variables $V_j$, $j=1,\ldots,N$ $F_V$ satisfies the following conditions
\begin{enumerate}
\item[(i)]
the probability distribution is absolutely continuous 
\begin{equation}
dF_V(v)=\rho(v)dv, \quad \supp\rho=[0,\ell]\\
\end{equation}
\item[(ii)] there exist $\rho_*$, $\overline{\rho}\in(0,\infty)$ such that
\begin{equation} \label{eq:rho.t}
\forall\in[0,\ell], \quad \rho_*\leq \rho(t)\leq \overline{\rho},\\
\end{equation}
\item[(iii)] $\rho$ has bounded  derivative  on  $(0,\ell)$
\begin{equation}
\|\rho'(\cdot)\Bone_{(0,\ell)}\|_{\infty}\leq C'_{\rho}
\end{equation}
\end{enumerate}
Then there exists $c_*=c_*(F_V)\in(0,\infty)$ such that for any $\delta\in(0,c_*N^{-3/2}]$
\[
\prob\{\nu_N(s)\geq \delta^{-1}s\}\leq \frac{4\overline{\rho}^2N^2\delta^2}{\ell^2}
\]
in particular, with $\delta=s^{\alpha}\leq C_*^{1\alpha} N^{-3/(2\alpha)}$, $\alpha\in(0,1)$, we have that
\[
\prob\{\nu_N(s)\geq s^{1-\alpha}\}\leq \frac{4\overline{\rho}^2}{\ell^2} N^2 s^{2\alpha}
\]
consequently, the i.i.d. random fields satisfying (i)--(iii), is  of the form $\condrcm$.
\end{theorem}

\begin{proof}[Step 1. Smoothness of the conditional measure]
By the smoothness assumption  (iii), the product probability measure with density
\[
p(x_1,\ldots,x_n)=\prod_{j=1}^n \rho(x_j)=\ee^{\sum_{j=1}^n \ln \rho(x_j)}
\]
induces on the interval $\CX(Y)\subset\CL(Y)$ a measure with smooth density with respect to  the Lebesgue measure on the line $\CL(Y)\subset\DR^N$. let $t=\tilde{\xi}_N$ be the normalized length parameter along $\CL(Y)$, then (cf. \eqref{eq:xi.tilde}).
\[
\CL(Y)=\{(\eta_1+t N^{-1/2},\ldots,\eta_N+ tN^{-1/2}), t\in\DR\}
\]
So the density at the point $t$ has the form
\[
p(t)=Z^{-1}(Y)\prob_{j=1}^np(\eta_j+t)=\ee^{\sum{j=1}^n\ln\rho(\eta_j+t)}
\]
where $Z^{-1}(Y)$ is the normalization factor. In particular,
\begin{equation}_label{eq:p.identity}
\frac{d}{dt} p(t)= N^{-1/2}p(t)\sum_{j=1}^N\frac{p'(\eta_j+tN^{-1/2})}{p(\eta_j+tN^{-1/2}}
\end{equation}

\textbf{Step 2 From $\nu$ to $|\CX(Y)|$} By \eqref{eq:p.identity} combined with assumption  \eqref{eq:rho.t}
\[
\|\frac{p'}{p}\mid_{\CX(Y)}\|_{\infty}\leq N\cdot N^{-1/2}C'_{\rho}\rho_*^{-1}\leq C_1 N^{1/2}
\]
In particular
\begin{equation}\label{eq:rho.prime}
\|p'\mid_{\CX(Y)}\|\leq C_1 N^{1/2}\|p\mid_{Y}\|_{\infty}
\end{equation}
For notational convenience, we identify $\CL(Y)$ with the real line $\DR$, equipped with the normalized coordinate $t=\tilde{\xi}_N$ and  let $t^*=t^*(Y)$ be any point of maximum of the density $\rho$ restricted to $\CX(Y)$ and  $\rho^*(Y)=\rho(t^*)$, the existence of $t^*(Y)$ follows from the continuity of $\rho$. Assume that
\[
|\CX(Y) |\geq 2\ell_N, \quad \ell_N\leq \ell_*N^{-1}
\]
where $\ell_*=\ell_*(F_V)\in(0,\infty)$ is small enough
\[
\ell_*(F_V)=(C_1(F_V))^{-1}
\]  
and depends upon the minimum of the density $p(\cdot)$ and the sup-norm of  its derivative, both of these quantities are determined by the probability distribution  function  $F_V$. Since  $|\CX(Y)|\geq 2\ell_N$ at least one of the interval.
$[t^*-\ell_N,t^*]$, $[t^*,t^*+\ell_N]$ (perhaps, both of them) is inside the interval $\CX(Y)$ and denote by $J_*$ such an intervals then for any $t\in\CX(Y)$, we have by \eqref{eq:rho.prime}
\begin{align*}
|p(t)-p(t^*)|&\leq \ell_* N^{-1}\cdot\max_{s\in J_*}\rho'(s)\\
(C_1\ell_*) N^{1/2}\cdot N^{-1}\cdot\rho^*(Y)
\end{align*}
So that $\forall t\in\CX(Y)$ and, e.g., $N\geq 4$,
\begin{align*}
\frac{1}{2}p^*(Y)&\leq p^*(Y)(1-N^{-1/2})\leq p(t)\\
&\leq p^*(Y)\cdot (1+N^{-1/2})\leq 2 p^*(Y)\\
\end{align*}  
the conditional measure induced, on $\CX(Y)$ has the form $dp_{Y}(t))=Z^{-1}(Y)p(t)dt$
with $Z(Y)=\int_{\CX(Y)} p(t)dt$, and we have 
\[
Z(Y)\geq \int_{J_*} p(t)dt=\frac{1}{2} p(t^*)\ell_N
\]
Therefore the assumption $|\CX(Y)|\geq 2\ell_N$, we have for any $t'\in\DR$:
\begin{align*}
\prob\{\xi_N\in[t',t'+s]\mid Y\}&=\prob\{\tilde{\xi}_N\in[t'',t''+N^{1/2}_s]\mid Y\}\\
&= Z^{-1}(Y)\int_{t''}^{t''+N^{1/2}s} p(t) dt\\
&\frac{p(t^*)N^{1/2}s}{p(t^*)\ell_N/2}=\frac{2N^{1/2}s}{\ell_N}
\end{align*}
(here $t''=N^{1/2}t'$), yielding for such  
\[
\nu_N(s\mid Y)\leq 2 N^{1/2}\ell^{-1}_Ns
\]
Therefore
\[
\{ \nu_N(s)\geq 2N^{1/2}\ell_N^{-1}s\}\subset\{|\CX(Y)|\leq 2\ell_N\}
\]
set $\delta:=\frac{1}{2}N^{-1/2}\ell_N$, $c_*=c_*(F_V):=\frac{1}{2}\ell_*(F_V)$. Then for any $\delta\in(0,c_*N^{-3/2}]$ 
\begin{equation}\label{eq:Nu.Ns.delta}
\{\nu_N(s)\geq \delta^{-1}\}\subset \{|\CX(Y)|\leq 4N^{1/2}\delta\}
\end{equation}
\textbf{Step 3. Conclusion}   
Now we apply Lemma \ref{lem:prob.densities}
\[
\prob\{|\CX(Y)|\leq r\}\leq \frac{1}{4} \overline{\rho}^2r^2 N,
\]
and obtain with $r=4N^{1/2}\delta$
\[
\prob\{ |\CX(Y)|\leq 4N^{1/2}\delta\}\leq 4\overline{p}^2N^2\delta^2
\] 
Now the main assertion follows from \eqref{eq:Nu.Ns} and \eqref{eq:Nu.Ns.delta}  for $\delta\in(0,c_*N^{-3/2})$
\begin{equation}\label{eq:Nu.Ns}
\prob\{\nu_N(s)\geq \delta^{-1} s\}\leq 4\overline{p}^2N^2\delta^2
\end{equation}
\end{proof}

\bibliographystyle{plain}

\end{document}